\documentclass[11pt]{article}
\usepackage{amsmath, amssymb, amscd, amsthm, amsfonts}
\usepackage{graphicx}
\usepackage{hyperref}
\usepackage{tikz}
\usepackage{bookmark}
\usepackage{tabularx}
\usepackage{listings}
\usepackage{algorithm}
\usepackage{algpseudocode}
\usepackage{titling}
\usepackage{lineno, blindtext}
\usetikzlibrary{external}
\usepackage{comment}
\usepackage{longtable}
\title{\fontsize{19}{21}\selectfont Computation of Graph Polynomials via Tree Decomposition: Theory, Algorithms, and Python Implementation}
\author{
    \fontsize{15}{17}\selectfont Mehul Bafna\textsuperscript{*} and Shaghik Amirian\textsuperscript{*} \\
    \\
    University of Applied Sciences Mittweida \\
     \textsuperscript{*}The authors contributed equally to this work.\\
}

\date{\today}

\newtheorem{theorem}{Theorem}

\newtheorem{definition}[theorem]{Definition}
\newtheorem{lemma}[theorem]{Lemma}

\begin{document}

\maketitle

\begin{abstract}
Graph polynomials encode fundamental combinatorial invariants of graphs. Their computation is investigated using tree and path decomposition frameworks, with formal definitions of treewidth, k-trees, and pathwidth establishing the structural basis for algorithmic efficiency. Explicit algorithms are constructed for each polynomial, leveraging decomposition order and state transformation mappings to enable tractable computation on graphs of bounded treewidth. Python implementations validate the methods, and computational complexity is analyzed with respect to sparse and k-degenerate graph classes. These results advance decomposition-based approaches for polynomial computation in algebraic graph theory.
\end{abstract}

\begin{center}
\small{\textbf{Keywords:} Graph Theory, Graph Polynomials, Graph Algorithms, Networks}
\end{center}

\section{Introduction}\label{section-introduction}

Graph theory plays a pivotal role in diverse areas, such as the creation of distance-based algorithms,
network 
flow, and traffic management-based applications. Graphs are a critical source for
holding information in the form of nodes and edges. With the plethora of information, it is 
vital to shorten the graph with the information being preserved. This study provides an
overview of the concept of tree decomposition that consolidates a simple undirected graph but still
preserves the general information and properties of that original graph. The concept of tree decomposition
was introduced by Neil Robertson and Paul Seymour \cite{ROBERTSON198339}. Furthermore, there are numerous applications of graphs with bounded tree-width that are utilized for the construction of significant algorithms in graph theory \cite{bodlaender2008combinatorial}. In the following manner, the remaining paper has been structured. Section 2 gives an overview of the tree decomposition, tree-width, k-trees,
and partial k-trees \cite{Bodlaender_1993,BODLAENDER19981}. Section 3 describes path decomposition, path-width, nice-path decomposition, and composition order \cite{ponitz2003methode}. Section 4 explains the computation of various graph polynomials,
namely the Independence Polynomial \cite{ind}, Chromatic Polynomial \cite{birkhoff1912determinant}, Domination Polynomial \cite{dom}, and the Bipartition Polynomial \cite{dod2015bipartition} based on an algorithmic approach on the composition order \cite{poly}.
Section 5 concludes the study with the implementation of the independence polynomial and domination
polynomial in Python. Moreover, the definitions in the paper have been studied from \cite{poly}.

\section{Tree decomposition}\label{section- tree decomposition}
\begin{definition}
For a given graph G = (V, E) with V and E as the vertex set and edge set, respectively, \textbf{tree decomposition} of G is described as a pair (T,$\gamma$) comprising 

- a tree T = (U,H) \

- a mapping $\gamma$ : U $\xrightarrow[]{}$ $2^{V}$ which renders to any node u of T a vertex subset $\gamma$(u) $\subseteq$ V \

 \hspace{2mm} satisfying below given three properties 
\begin{enumerate}
    \item (Vertex preservation) Every vertex of G is allocated to at least one node of T, which gives the relation 
    \ 
    
    \[ \bigcup_{u \in U} \gamma (u) = V\]
   
    \item (Edge preservation) For every edge $e \in$ E, there exists a node u in T consisting of the end vertices of $e$, i.e.
    \ 
    
    \[ \forall \hspace{1mm} e \in E: \exists \hspace{1mm} u \in U: e \subseteq \gamma (u)  \]

    \item (Compactness) For nodes $u_1, u_2$ and $u_3$ in T, if $u_2$ exists in unique path between $u_1$ and $u_3$, below relation holds 
    \ 
    
    \[ \gamma (u_1) \cap  \gamma (u_3) \subseteq \gamma (u_2)\] 
\end{enumerate}
\end{definition} 

\ 

\begin{figure}[h!]
    \centering
    \begin{tikzpicture}  
        [scale=1.8,auto=center,every node/.style={circle,fill=transparent!20}]  
    
        \node (a1) at (-1.732/2,1/2) {$v_1$};  
        \node (a2) at (0,0)  {$v_2$};  
        \node (a3) at (-1.732/2,-1/2)  {$v_3$};  
        \node (a4) at (1,0) {$v_4$};  
        \node (a5) at (3.732/2,1/2)  {$v_5$};  
        \node (a6) at (3.732/2,-1/2)  {$v_6$};  
        \node (a7) at (2.75,0)  {$v_7$}; 
    
        \draw (a1) -- (a2);  
        \draw (a2) -- (a3);  
        \draw (a1) -- (a3);
        \draw (a2) -- (a4);
        \draw (a5) -- (a4);
        \draw (a6) -- (a4);
        \draw (a7) -- (a5);
        \draw (a7) -- (a6);
    \end{tikzpicture}  
    
    \caption{Graph representation with 7 vertices and 8 edges}
    \label{fig:graph_example}
\end{figure}
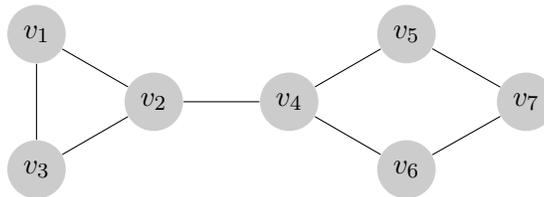

\begin{table}[h!]
\centering

\begin{tabularx}{1.0\textwidth} { 
  | >{\centering\arraybackslash}X 
  | >{\centering\arraybackslash}X 
  | }

 \hline
 $T_i$  $|$ $ \forall \hspace{1mm} i \in \{1,2\}$ & Nodes \\
 \hline

\begin{center}
     \begin{tikzpicture}  
  [scale=1.8,auto=center,every node/.style={circle,fill=transparent!20}]  
  \node (a1) at (0,0) {$v_1, v_2, v_3$};  
  \node (a2) at (0,-1)  {$v_2, v_4$};  
  \node (a3) at (-1.732/2,-3/2)  {$v_4, v_5$};  
  \node (a4) at (1.732/2,-3/2) {$v_4, v_6$};  
  \node (a5) at (-1.732/2,-5/2)  {$v_5, v_7$};  
  \node (a6) at (1.732/2,-5/2)  {$v_6,v_7$};  

  \draw (a1) -- (a2);  
  \draw (a3) -- (a2);
  \draw (a4) -- (a2);
  \draw (a3) -- (a5);
  \draw (a4) -- (a6);
\end{tikzpicture}  
\end{center}

&
\vspace{25mm}
\begin{center}
    \{\{$v_1, v_2,v_3$\}, \{$v_2,v_4$\}, \{$v_4,v_5$\},  \{$v_4,v_6$\},  \{$v_5,v_7$\},\{$v_6,v_7$\}\}
\end{center}

\\
\hline

\begin{center}
    \begin{tikzpicture}  
  [scale=1.8,auto=center,every node/.style={circle,fill=transparent!20}]  
  \node (a7) at (0,0) {$v_1, v_2, v_3$};
  \node (a8) at (0,-1) {$v_2, v_4$};
  \node (a9) at (0,-2) {$v_4, v_5,v_6$};
  \node (a10) at (-0.866, -2.5) {$v_5,v_7$};
  \node (a11) at (0.866, -2.5) {$v_6,v_7$};

  \draw (a7) -- (a8);
  \draw (a8) -- (a9);
  \draw (a10) -- (a9);
  \draw (a11) -- (a9);
\end{tikzpicture} 
\end{center}

& 

\vspace{25mm}
\begin{center}
    \{\{$v_1, v_2,v_3$\}, \{$v_2,v_4$\}, \{$v_4,v_5,v_6$\}, \{$v_5,v_7$\},\{$v_6,v_7$\}\} 
\end{center}

\\
\hline

\end{tabularx}

\caption{Two distinct tree decompositions for graph in Figure 1 with its respective nodes}
\label{tab:tree-decompositions}

\end{table}

\subsection{Treewidth}\label{subsection-Treewidth}

\begin{definition}
Given a graph G = (V, E) with (T = (U, H),$\gamma$) as its tree decomposition, then the width of the tree decomposition is described as 
\[ w(T,\gamma) = max \{ |\gamma (y)| \hspace{1mm} | y \in U \} - 1\]

\

\textbf{Treewidth}, denoted as tw(G), is defined as the minimum width of a tree decomposition of G.
\end{definition}

\begin{figure}[h!]
\centering
\begin{tikzpicture}  
  [scale=1.8,auto=center,every node/.style={circle,fill=transparent!20}]
  
  \node (a1) at (-1.732/2,1/2) {$v_2$};  
  \node (a2) at (0,0)  {$v_4$};  
  \node (a3) at (-1.732/2,-1/2)  {$v_3$};  
  \node (a4) at (1,0) {$v_5$};  
  \node (a5) at (3.732/2,1/2)  {$v_6$};  
  \node (a6) at (3.732/2,-1/2)  {$v_7$};  
  \node (a7) at (2.75,0)  {$v_8$}; 
  \node (a8) at (-1.732,0)  {$v_1$}; 

  \draw (a1) -- (a2);  
  \draw (a2) -- (a3);  
  \draw (a1) -- (a8);
  \draw (a3) -- (a8);  
  \draw (a2) -- (a4);
  \draw (a5) -- (a4);
  \draw (a6) -- (a4);
  \draw (a7) -- (a5);
  \draw (a7) -- (a6);
\end{tikzpicture}

\caption{Graph representation with 8 vertices and 9 edges}
\label{fig:custom-graph}
\end{figure}
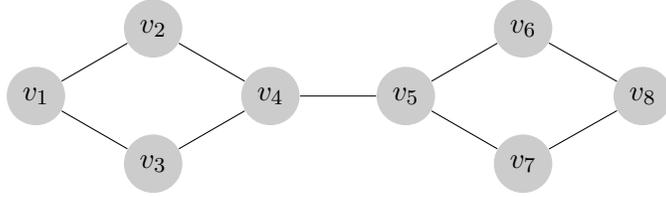

\begin{table}[h!]

\begin{center}
  \begin{tabularx}{0.5\textwidth} { 
  | >{\centering\arraybackslash}X 
  | >{\centering\arraybackslash}X 
  | }
 \hline

$T_1$  \\
\hline
 \begin{center}
   \begin{tikzpicture}  
  [scale=1.8,auto=center,every node/.style={circle,fill=transparent!20}] 
  \node (a10) at (5.5,0) {$v_1,v_2,v_3$};
  \node (a11) at (5.5,-1) {$v_2,v_3,v_4$};
  \node (a12) at (5.5,-2) {$v_4,v_5$};
  \node (a13) at (4.793,-2.707) {$v_5,v_6$};
  \node (a14) at (6.207,-2.707) {$v_5,v_7$};
  \node (a15) at (4.793,-3.707) {$v_6,v_8$};
  \node (a16) at (6.207,-3.707) {$v_7,v_8$};
  
  \draw (a10) -- (a11);
  \draw (a11) -- (a12);
  \draw (a12) -- (a13);
  \draw (a12) -- (a14);
  \draw (a13) -- (a15);
  \draw (a16) -- (a14); 
  
\end{tikzpicture}  
\end{center}\\
\hline
$w(T_1,\gamma_1)= max \{ |\gamma_1 (y_1)| \hspace{1mm} | y_1 \in U_1 \}$ - 1 = 2\\
\hline

\end{tabularx}  
\end{center}

\caption{Graph in figure 2 with tw=1}
\label{tab:tree-decompositions}

\end{table}

\subsection{K-trees and partial k-trees}\label{subsection-K-trees and partial k-trees}

A \textbf{k-tree} is constructed with the following procedure:
\begin{enumerate}
\item Initially, to form a complete graph with order k+1, as it is isomorphic to a k-tree.
\item On a recursive basis, adding vertices in a way that every added vertex v follows:
\begin{itemize}
\item $|$ N(v) $|$ = k
\item The graph induced by N[v] forms a (k+1)-clique
\end{itemize}
\end{enumerate}

A \textbf{partial k-tree} is a graph that can be derived from a k-tree by the removal of its edges. In other words, if we have a k-tree G = (V, E), then a partial k-tree is H = (V', E') such that V = V' and E' $\subseteq$ E. A partial k-tree is a spanning subgraph of the k-tree.

\begin{theorem}
A partial k-tree with order m has at most $km - \binom{k+1}{2}$ edges. 
\end{theorem}

\begin{proof}
Let G be a k-tree with order m. For computing the number of edges in G, we initially consider the $\binom{k+1}{2}$ edges of $K_{k+1}$. Now the remaining vertices in the k-tree are $m - (k+1)$, and as per the second part of the k-tree definition, the remaining edges are $(m - (k+1))\cdot k$. Hence, on summation we obtain $E(G)=km - \binom{k+1}{2}$. Since a partial k-tree is generated from a k-tree by the removal of edges, the number of edges for a k-tree serves as an upper bound for the edges of a partial k-tree.
\end{proof}

\begin{figure}[H]
\centering
\begin{tikzpicture}  
  [scale=6.5,auto=center,every node/.style={circle,fill=transparent!20}] 
  
  \node (a1) at (0,0) {$v_1$};
  \node (a2) at (1,0) {$v_2$};
  \node (a3) at (1/2,1.732/2) {$v_3$};
  \node (a4) at (1/2,0.2887) {$v_4$};
  \node (a5) at (1/3,0.3849) {$v_5$};
\node (a6) at (2/3,0.3849) {$v_6$};
\node (a7) at (1/2,0.096) {$v_7$};
\node (a8) at (1/2,1.1547) {$v_8$};
  \draw (a1) -- (a2);
  \draw (a1) -- (a3);
  \draw (a1) -- (a4);
  \draw (a2) -- (a3);
  \draw (a3) -- (a4);
  \draw (a2) -- (a4);
  \draw (a1) -- (a5);
  \draw (a3) -- (a5);
  \draw (a4) -- (a5);
  \draw (a2) -- (a6);
 \draw (a3) -- (a6);
 \draw (a4) -- (a6);
\draw (a1) -- (a7);
\draw (a2) -- (a7);
\draw (a4) -- (a7);
\draw (a1) -- (a8);
\draw (a2) -- (a8);
\draw (a3) -- (a8);
  
\end{tikzpicture}

\caption{A 3-tree with 8 vertices}
\end{figure}
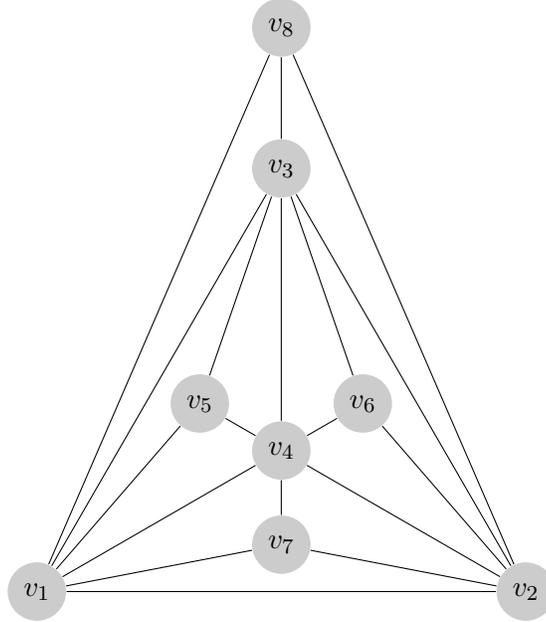

\section{Path decomposition}\label{section-Path decomposition}

A path decomposition is a special case of tree decomposition, such that T in tree decomposition $(T,\gamma)$ is a path. 

\begin{definition}
For a given graph G = (V, E) with V and E as the vertex set and edge set, respectively, the \textbf{path decomposition} is defined as a sequence $(V_1,V_2,\cdots,V_k)$ of the vertex subsets (not necessarily disjoint) of G that satisfies

\begin{enumerate}
    \item (Vertex preservation) For any vertex $v \in V$, there exist at least one vertex subset $V_j$  with $j \in \{1,\cdots,k\}$ that contains v
    \[ \bigcup_{j \in \{1,\cdots,k\}} V_j = V\]
   
    \item (Edge preservation) Every edge $e = \{u,v\}\in$ E, is contained in one of the vertex subset $V_j$ with \ 
    
    $j \in \{1,\cdots,k\}$ 
    \[ \forall \hspace{1mm} e \in E: \exists \hspace{1mm} j \in \{1,\cdots,k\}: u,v \in V_j  \]

    \item (Compactness) $V_a \cap V_c \subseteq V_b \hspace{2mm} \forall \hspace{1mm} a,b,c \in \{1,\cdots,k\}$ with $1 \leq a < b <c \leq k $
    
\end{enumerate}
\end{definition}

\subsection{Pathwidth}
\begin{definition}
For a given graph G = (V, E), \textbf{pathwidth} denoted as pw(G) is defined as the minimum width for a path decomposition of G. \\ 

\end{definition}
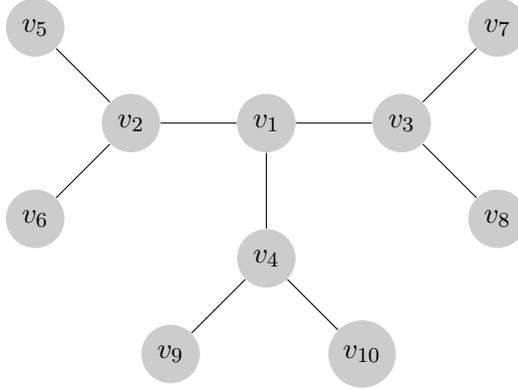
\begin{figure}[H]
    \centering
    \begin{tikzpicture}  
  [scale=1.8,auto=center,every node/.style={circle,fill=transparent!20}] 
    
  \node (a1) at (0,0) {$v_1$};  
  \node (a2) at (-1,0)  {$v_2$}; 
  \node (a3) at (1,0)  {$v_3$};  
  \node (a4) at (0,-1) {$v_4$};  
  \node (a5) at (-1.707,0.707)  {$v_5$};  
  \node (a6) at (-1.707,-0.707) {$v_6$};  
  \node (a7) at(1.707,0.707)  {$v_7$}; 
  \node (a8) at(1.707,-0.707)  {$v_8$}; 
  \node (a9) at(-0.707,-1.707)  {$v_9$}; 
  \node (a10) at(0.707,-1.707)  {$v_{10}$}; 
  
  \draw (a1) -- (a2); 
  \draw (a1) -- (a4);  
  \draw (a1) -- (a3);
  \draw (a2) -- (a6);
  \draw (a5) -- (a2);
  \draw (a3) -- (a8);
  \draw (a7) -- (a3);
  \draw (a4) -- (a9);
  \draw (a4) -- (a10);

\end{tikzpicture}  

    \caption{Graph representation with 10 vertices and 9 edges}
    \label{fig:placeholder}
\end{figure}

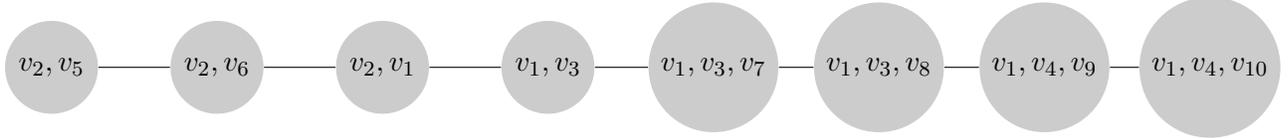
\begin{figure}[H]
\centering
    \begin{tikzpicture}  
  [scale=3.3,auto=center,every node/.style={circle,fill=transparent!20}] 
    
  \node (a1) at (-8/3,0) {$v_2,v_5$};  
  \node (a2) at (-2,0)  {$v_2,v_6$}; 
  \node (a3) at (-4/3,0)  {$v_2,v_1$};  
  \node (a4) at (-2/3,0) {$v_1,v_3$};  
  \node (a5) at (0,0)  {$v_1,v_3,v_7$};  
  \node (a6) at (2/3,0) {$v_1,v_3,v_8$};  
  \node (a7) at (4/3,0)  {$v_1,v_4,v_9$}; 
  \node (a8) at (2,0)  {$v_1,v_4,v_{10}$};  
  
  \draw (a1) -- (a2); 
  \draw (a2) -- (a3);  
  \draw (a4) -- (a3);
  \draw (a4) -- (a5);
  \draw (a5) -- (a6);
  \draw (a6) -- (a7);
  \draw (a7) -- (a8);

\end{tikzpicture}  
\caption{Path decomposition for the graph in Figure 4 with pw = 2}
\end{figure}

\subsection{Nice path decomposition and composition order}\label{section-nice path decomposition}
\subsubsection{Nice path decomposition}
\begin{definition}
For a given graph $G = (V, E)$, a path decomposition ($V_1,\cdots,V_k$) is defined as a \textbf{nice path decomposition} if it satisfies the properties mentioned in Definition 4 in addition to the following properties - \

\begin{enumerate}
\item $V_1 = V_k = \phi$
\item Any two consecutive vertex subsets differ by exactly one element, i.e. 
\[ |V_j \setminus V_i| = 1 \hspace{1mm} \text{with} \hspace{1mm}i,\hspace{0.2mm}j \in \{ 1,\cdots,k \} \hspace{1mm} and \hspace{1mm} |j-i|=1\] 
\end{enumerate}
\end{definition}

For the given graph $G = (V, E)$ with order $m$, nice path decomposition can also be described in terms of \ 

signed vertices as sequence ($\alpha_1u_1,\cdots,\alpha_{2m}u_{2m}$), where $u_j \in V$ for $j \in \{1,\cdots,2m\}$ and $\alpha_j$  is described as:

   \[\alpha_j=
    \begin{cases}
      \hspace{2.5mm} 1, & \text{if}\ V_j \setminus V_{j-1} = \{u_j\}, \\
      -1, & \text{if}\ V_{j-1} \setminus V_j =\{u_j\}
    \end{cases}\]

where $j \in \{1,\cdots,2m-1\}$ and $|V_j \setminus V_{j-1}|$ denotes the difference between two consecutive vertex subsets.

\subsubsection{Composition order}

\begin{definition}
For a given graph G = (V, E), a \textbf{composition order} is defined as a sequence ($y_1,\cdots,y_t$) consisting of edges and signed vertices satisfying the following properties -

\begin{enumerate}
    \item Every edge $e \in E$ occurs only once in the composition order
    \item After removal of every edge, the obtained sub-sequence is a nice path decomposition
    \item Suppose $y_i = \{f,g\}$ is an edge of G, then the composition order consists of four terms $y_a, y_b, y_c, y_d$ that satisfies $a<i, b<i, c>i$ and $d>i$ such that $y_a = +f$, $y_b = +g$, $y_c = -f$ and $y_d = -g$.
    \item t = 2$\cdot$order(G) + size(G)
\end{enumerate}
\end{definition}

\begin{figure}[H]
\centering
    \begin{tikzpicture}  
  [scale=1.8,auto=center,every node/.style={circle,fill=transparent!20}] 
    
  \node (a1) at (-1.732/2,1/2) {$v_1$};  
  \node (a2) at (0,0)  {$v_2$}; 
  \node (a3) at (-1.732/2,-1/2)  {$v_3$};  
  \node (a4) at (1,0) {$v_4$};  
  \node (a5) at (3.732/2,1/2)  {$v_5$};  
  \node (a6) at (3.732/2,-1/2)  {$v_6$};  

  \draw (a1) -- (a2); 
  \draw (a2) -- (a3);  
  \draw (a1) -- (a3);
  \draw (a2) -- (a4);
  \draw (a5) -- (a4);
  \draw (a6) -- (a4);
  \draw (a6) -- (a5);

\end{tikzpicture}  
\caption{Graph representation with 6 vertices and 7 edges}
\end{figure}
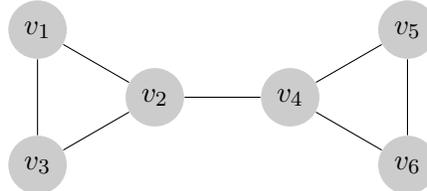

For the graph given in Figure 6, below are the nice path decomposition and composition order \\

\textbf{Nice path decomposition: } ($+v_1,+v_3,+v_2,-v_1,-v_3,+v_4,-v_2,+v_5,+v_6,-v_4,-v_5,-v_6$)
\ 

\textbf{Composition order: } ($+v_1,+v_3,\{v_1,v_3\},+v_2,\{v_1,v_2\},\{v_2,v_3\},-v_1,-v_3,+v_4,\{v_2,v_4\},-v_2,+v_5,$

\hspace{37mm}$\{v_4,v_5\},+v_6,\{v_4,v_6\},-v_4,\{v_5,v_6\},-v_5,-v_6$)

\section{Computation of graph polynomials based on composition order}

In this section, we discuss the construction of graph polynomials based on the composition order using an algorithmic approach \cite{poly}. The algorithm with respect to the path decomposition comprises a sequence of fundamental steps. The fundamental step resembles a change in a state set with every step. A state $z$ constitutes $y$ and $f$ as a pair, denoting an index and a value, respectively. Depending on the type of polynomial that is being computed, an index can be a set, a partition, or any other object varying in complexity. Since on termination of the algorithm, it is required to obtain the desired polynomial, so the value rendered is a polynomial. \\ 

 The factors involved in the composition order can be summarized as follows:

\begin{itemize}
    \item Vertex addition($+v$)
    \item Vertex deletion($-v$)
    \item Edge insertion($e$)
\end{itemize}

In the formulation of the algorithm, these factors are employed to determine the state transformation as mentioned below: \

-- \underline{State initialization}: In the primary step, the state is initialized with $Z_0 = (y_0,f_0)$ that will further be transformed to a new state relying on one of the above three factors. 

\

-- \underline{Vertex addition($+v$)} : This step is executed when a vertex with a positive(+) prefix is detected in the composition order sequence. The current state $z$ is transformed into a new state $z'$ after this step has been executed. $Z$ denotes set of current states and $Z'$ contains all correct existing described pairs ($y,f$) with $x$ as an index and $f$ an allowable value for the given problem. The vertex addition process is defined as a mapping 

\[ \delta_{v} : Z \xrightarrow[]{} 2^{Z}\]

the transformed set of states $Z'$ in terms of the preceding set of states $Z$ is shown as

\[ Z' = \bigcup_{z\in Z} \delta_{v}(z) \] \

-- \underline{Vertex deletion($-v$)} : This step is executed when a vertex with a negative(-) prefix is detected in the composition order sequence. The current state $z$ is transformed to a new state $z'$ after this step has functioned. $Z$ denotes set of current states and $Z'$ contains all correct existing described pairs ($y,f$) with $x$ as an index and $f$ an allowable value for the given problem. The vertex deletion process is defined as a mapping 

\[ \eta_{v} : Z \xrightarrow[]{} 2^{Z}\]

the transformed set of states $Z'$ in terms of the preceding set of states $Z$ is shown as

\[ Z' = \bigcup_{z\in Z} \eta_{v}(z) \] \

-- \underline{Edge insertion($e$)} : This step is executed while processing an edge with endpoints $u$ and $v$ in the composition order sequence. The current state $z$ is transformed into a new state $z'$ after this step has been executed. $Z$ denotes set of current states and $Z'$ contains all correct existing described pairs ($y,f$) with $x$ as an index and $f$ an allowable value for the given problem. The edge insertion process is defined as a mapping \[ \beta_e : Z \xrightarrow[]{} 2^{Z}\]
the transformed set of states $Z'$ in terms of the preceding set of states $Z$ is shown as

\[ Z' = \bigcup_{z\in Z} \beta_e(z) \] \

There is a possibility of an instance where multiple states $z_1, z_2$ have an identical index on 
performing any of the above three steps. Depending on the polynomial type that is computed, such multiple states are replaced with a single state having that common index.\\

\textbf{Remark} : The three maps $\delta_{v}$, $\eta_{v}$ and $\beta_e$ varies according to the polynomial

\subsection{Independence polynomial}

To delve into the concept of independence polynomial \cite{ind}, it is required to go through the notion of independent sets. For a given graph $G = (V, E)$, a vertex subset $V' \subseteq V$ is termed an independent set if the graph induced by $V'$ that is G[$V'$] is empty. The independent set with the maximum cardinality is defined as the maximum independent set, and the cardinality of such a maximum independent set is described as the independence number, denoted by $\alpha(G)$.  \

\begin{definition}
For a given graph $G = (V, E)$, the independence polynomial in terms of independent sets and independence number is defined as \[ I(G,x) = \sum_{j=0}^{\alpha(G)} m_j(G)x^{j} \] \ 
where the coefficients $m_j(G)$ gives the total number of independent sets with cardinality $j$,  $\forall \hspace{1mm} j \in \{0,\cdots,\alpha(G)\}$.\
\end{definition}

For the graph mentioned in Figure 6, below are the values for the maximum independent set(s), \ 

independence number and independence polynomial. \\

\textbf{Maximum independent sets}: $\{v_1,v_4\},\{v_1,v_5\},\{v_1,v_6\},\{v_3,v_4\},\{v_3,v_5\},\{v_3,v_6\},\{v_2,v_5\},\{v_2,v_6\}$ \\

\textbf{Independence number}: 2 \\ 

\textbf{Independence polynomial} : $1+6x+8x^{2}$ \

\begin{lemma}
For a graph G with order n, $n+1 \leq nis(G) \leq 2^{n}$, where nis(G) denotes the total number of independent sets of G.
\end{lemma}

\begin{proof}
For obtaining the lower and upper bounds for $nis(G)$, the class of graphs with the most and least number of edges, i.e., complete graphs($K_n$) and empty graphs($E_n$) respectively are considered. Since in $K_n$, every vertex is adjacent to every other vertex, it gives the lower bound for $nis(G)$ as $n+1$. In the case of $E_n$, since there are no edges, so any set of vertices with cardinality in the range [$0,n$] are considered that renders the power set returning the upper bound as $2^{n}$. 
\end{proof}

For the computation of independence polynomial based on above mentioned algorithmic approach, below are the required mappings :\\

\begin{table}[H]
\begin{tabular}{|p{4cm}|p{11cm}|}
\hline
\begin{center}
     Initial state ($Z_0$) 
\end{center}& \begin{center}
   $\{(\phi,1)\}$ 
\end{center} \\
\hline
\begin{center}
    Vertex addition(+$v$) 
\end{center}& \begin{center}
    $\delta_{v}((y,f)) = \{(y,f),(y+v,xf)\} $
\end{center}   \\
\hline
\begin{center}
    Vertex deletion(-$v$) 
\end{center}& \begin{center}
    $\eta_{v}((y,f)) = \{(y-v,f)\} $ 
\end{center}  \\
\hline
\vspace{6mm}

\begin{center}
    Edge insertion($e$)
\end{center}

    &
    
    \begin{center}
        $\beta_{e}((y,f)) = 
    \begin{cases}
      (y,f), & \text{if}\ e = \{u,v\} \not\subseteq y \\
      \phi, & \text{otherwise}
    \end{cases} $

    \end{center}

 \vspace{9mm}
 
    \\
\hline
\end{tabular}
\caption{Maps for computation of independence polynomial}
\end{table}

\textbf{Explanation}: $y + v$ denotes $y \cup\{v\} $ and $y- v$ represents deletion of vertex from the set. Indices are in the form of sets. Since the independence polynomial of the null graph is 1, thereby gives the initialized state as ($\phi,1$). Two states $z_1,z_2$ having identical index $y$ given by $(y,f_1)$ and $(y,f_2)$ respectively is replaced by $(y,f_1+f_2)$. In the composition order sequence, while processing an edge, we have two cases as by the definition of independent sets, there exists no edge between any two vertices in the independent set. Hence, if an edge is a subset of the value $y$, those particular sets containing that edge are removed, and the remaining are retained.

\begin{table}[H]
\begin{tabular}{|p{1.5cm}|p{14cm}|}
\hline
Step & States \\
\hline
init & ($\phi$,1) \\
\hline
+$v_1$ & ($\phi$,1),(\{$v_1$\},$x$)\\
\hline
+$v_3$ & ($\phi$,1),(\{$v_1$\},$x$),(\{$v_3$\},$x$),(\{$v_1,v_3$\},$x^{2})$\\
\hline
$\{v_1,v_3\}$ & ($\phi$,1),(\{$v_1$\},$x$),(\{$v_3$\},$x$)\\
\hline
+$v_2$ & ($\phi$,1),(\{$v_1$\},$x$),(\{$v_3$\},$x$),(\{$v_2$\},$x$),(\{$v_1,v_2$\},$x^{2}$),(\{$v_2,v_3$\},$x^{2})$\\
\hline
$\{v_1,v_2\}$ & ($\phi$,1),(\{$v_1$\},$x$),(\{$v_3$\},$x$),(\{$v_2$\},$x$),(\{$v_2,v_3$\},$x^{2}$)\\
\hline
$\{v_2,v_3\}$ & ($\phi$,1),(\{$v_1$\},$x$),(\{$v_3$\},$x$),(\{$v_2$\},$x$)\\
\hline
-$v_1$ & ($\phi$,$1+x$),(\{$v_3$\},$x$),(\{$v_2$\},$x$)\\
\hline
-$v_3$ & ($\phi$,$1+2x$),(\{$v_2$\},$x$)\\
\hline
+$v_4$ & ($\phi$,$1+2x$),(\{$v_2$\},$x$),(\{$v_4$\},$x+2x^{2}$),(\{$v_2,v_4$\},$x^{2})$\\
\hline
$\{v_2,v_4\}$ & ($\phi$,$1+2x$),(\{$v_2$\},$x$),(\{$v_4$\},$x+2x^{2})$\\
\hline
-$v_2$ & ($\phi$,$1+3x$),(\{$v_4$\},$x+2x^{2}$)\\
\hline
+$v_5$ & ($\phi$,$1+3x$),(\{$v_4$\},$x+2x^{2}$),(\{$v_5$\},$x+3x^{2}$),(\{$v_4,v_5$\},$x^{2}+2x^{3}$)\\
\hline
$\{v_4,v_5\}$ & ($\phi$,$1+3x$),(\{$v_4$\},$x+2x^{2}$),(\{$v_5$\},$x+3x^{2}$)\\
\hline
+$v_6$ & ($\phi$,$1+3x$),(\{$v_4$\},$x+2x^{2}$),(\{$v_5$\},$x+3x^{2}$),(\{$v_6$\},$x+3x^{2}$),(\{$v_4,v_6$\},$x^{2}+2x^{3}$),(\{$v_5,v_6$\},$x^{2}+3x^{3}$)\\
\hline
$\{v_4,v_6\}$ & ($\phi$,$1+3x$),(\{$v_4$\},$x+2x^{2}$),(\{$v_5$\},$x+3x^{2}$),(\{$v_6$\},$x+3x^{2}$),(\{$v_5,v_6$\},$x^{2}+3x^{3}$)\\
\hline
-$v_4$ & ($\phi$,$1+4x+2x^{2}$),(\{$v_5$\},$x+3x^{2}$),(\{$v_6$\},$x+3x^{2}$),(\{$v_5,v_6$\},$x^{2}+3x^{3}$)\\
\hline
$\{v_5,v_6\}$ & ($\phi$,$1+4x+2x^{2}$),(\{$v_5$\},$x+3x^{2}$),(\{$v_6$\},$x+3x^{2}$)\\
\hline
-$v_5$ & ($\phi$,$1+5x+5x^{2}$),(\{$v_6$\},$x+3x^{2}$)\\
\hline
-$v_6$ & ($\phi$,$\boldsymbol{1+6x+8x^{2}}$)\\
\hline
\end{tabular}
\caption{Computation of independence polynomial for graph in Figure 6}
\end{table}

\subsection{Chromatic polynomial}

Chromatic polynomial was introduced by George David Birkhoff \cite{birkhoff1912determinant}. Following are some relevant definitions w.r.t. the chromatic polynomial.

\begin{definition}
For a given graph $G = (V,E)$, a \textbf{vertex coloring} is defined as a mapping $Z : V \xrightarrow[]{} C$ that allots a color to vertex from $C$ a finite set of colors. In addition, a coloring $Z$ is termed as \textbf{proper} if for any edge $e = \{a,b\} \in E$, $Z(a) \neq Z(b)$ $\forall \hspace{1mm} a,b \in V$.
\end{definition}

\begin{definition}
For a given graph $G = (V, E)$, \textbf{chromatic number} denoted by $\chi(G)$ is defined as the minimum number of colors needed to properly color all the vertices in the vertex set.
\end{definition}

For a given graph $G = (V,E)$, chromatic polynomial $P(G,x)$ is a unique polynomial with degree $|V|$ that gives the number of ways of properly coloring all vertices of G with a color set of cardinality x, where x $\in$ N.\\

For the graph mentioned in Figure 6 below are the values for chromatic number and chromatic polynomial.

\textbf{Chromatic number}: 3 \ 

\textbf{Chromatic polynomial} : $x^{6}-7x^{5}+19x^{4}-25x^{3}+16x^{2}-4x$ 

\begin{lemma}
For a graph G with order n, $\chi(G) \leq n -\alpha(G)+1$.
\end{lemma}

\begin{proof}
Since by the definition of independent sets, there exist no edges between any two vertices in that set. Hence, it gives the chromatic number for that set as 1. Now the remaining vertices can be properly colored by a maximum of $n-\alpha(G)$ colors if the graph induced by the remaining vertices is isomorphic to $K_{n-\alpha(G)}$. It gives the upper bound for the chromatic number $\chi(G)$ as $n -\alpha(G)+1$. 
\end{proof}

For the computation of chromatic polynomial based on above mentioned algorithmic approach, below are the required mappings:\\

\begin{table}[H]
\begin{tabular}{|p{3.5cm}|p{13cm}|}
\hline
\begin{center}
    Initial state ($Z_0$) 
\end{center}
 &  \begin{center}
     $\{(\phi,1)\}$
 \end{center}\\
\hline
\begin{center}
    Vertex addition(+$v$) 
\end{center}& \begin{center}
    
$\delta_{v}((y,f)) = \{(y|v,(x-|y|)f)\}\cup \{((y \smallsetminus B)\cup \{B\cup \{v\}\},f):B \in y\} $   \end{center}\\
\hline
\begin{center}
   Vertex deletion(-$v$)  
\end{center}& \begin{center}
   $\eta_{v}((y,f)) = \{(y-v,f)\} $  
\end{center}  \\
\hline
\vspace{6mm}

\begin{center}
    Edge insertion($e$)
\end{center}

    &
    \begin{center}
        $\beta_{e}((y,f)) = 
    \begin{cases}
      \phi, & \text{if there exists a block B $\in$ y}\ \text{and } e=\{u,v\} \subseteq B,  \\
      \{(y,f)\}, & \text{otherwise}
    \end{cases} $ 
    \end{center} 

\vspace{9mm}
 
    \\
\hline
\end{tabular}
\caption{Maps for computation of chromatic polynomial}

\end{table}

\textbf{Explanation}: Indices are considered in the form of a partition. Since the chromatic polynomial of a null graph is 1, thereby gives the initialized state as ($\phi,1$). Two states $z_1,z_2$ having identical index $y$ given by $(y,f_1)$ and $(y,f_2)$ respectively is replaced by $(y,f_1+f_2)$. While considering vertex addition, the vertex is added to each of the existing blocks and also to a separate block. In the case of edge insertion, two cases are considered. While processing an edge $e = uv$ (in shorter form), if $u$ and $v$ are present in the same block, then that partition is removed as per the definition of proper coloring.

\begin{longtable}{|p{1cm}|p{14cm}|}
\hline
Step & States \\
\hline
init & ($\phi$,1) \\
\hline
+$v_1$ & ($v_1$,$x$)\\
\hline
+$v_3$ & ($v_1|v_3$,$x^{2}-x$),($v_1v_3$,$x$)\\
\hline
$v_1v_3$ & ($v_1|v_3$,$x^{2}-x$)\\
\hline
+$v_2$ & ($v_1v_2|v_3$,$x^{2}-x$), ($v_1|v_3v_2$,$x^{2}-x$),($v_1|v_3|v_2$,$x^{3}-3x^{2}+2x$) \\
\hline
$v_1v_2$ & ($v_1|v_3v_2$,$x^{2}-x$),($v_1|v_3|v_2$,$x^{3}-3x^{2}+2x$) \\
\hline
$v_2v_3$ & ($v_1|v_3|v_2$,$x^{3}-3x^{2}+2x$) \\
\hline
-$v_1$ & ($v_3|v_2$,$x^{3}-3x^{2}+2x$)\\
\hline
-$v_3$ & ($v_2$,$x^{3}-3x^{2}+2x$)\\
\hline
+$v_4$ & ($v_2|v_4$,$x^{4}-4x^{3}+5x^{2}-2x$),($v_2v_4$,$x^{3}-3x^{2}+2x$)\\
\hline
$v_2v_4$& ($v_2|v_4$,$x^{4}-4x^{3}+5x^{2}-2x$)\\
\hline
-$v_2$ & ($v_4$,$x^{4}-4x^{3}+5x^{2}-2x$)\\
\hline
+$v_5$ & ($v_4|v_5$,$x^{5}-5x^{4}+9x^{3}-7x^{2}+2x$),($v_4v_5$,$x^{4}-4x^{3}+5x^{2}-2x$)\\
\hline
$v_4v_5$ & ($v_4|v_5$,$x^{5}-5x^{4}+9x^{3}-7x^{2}+2x$)\\
\hline
+$v_6$ & ($v_4v_6|v_5$,$x^{5}-5x^{4}+9x^{3}-7x^{2}+2x$),($v_4|v_5v_6$,$x^{5}-5x^{4}+9x^{3}-7x^{2}+2x$),\ 

($v_4|v_5|v_6$,$x^{6}-7x^{5}+19x^{4}-25x^{3}+16x^{2}-4x$)\\
\hline
$v_4v_6$ & ($v_4|v_5v_6$,$x^{5}-5x^{4}+9x^{3}-7x^{2}+2x$),
($v_4|v_5|v_6$,$x^{6}-7x^{5}+19x^{4}-25x^{3}+16x^{2}-4x$)\\
\hline
-$v_4$ & ($v_5v_6$,$x^{5}-5x^{4}+9x^{3}-7x^{2}+2x$),
($v_5|v_6$,$x^{6}-7x^{5}+19x^{4}-25x^{3}+16x^{2}-4x$)\\
\hline
$v_5v_6$ & ($v_5|v_6$,$x^{6}-7x^{5}+19x^{4}-25x^{3}+16x^{2}-4x$)\\
\hline
-$v_5$ & ($v_6$,$x^{6}-7x^{5}+19x^{4}-25x^{3}+16x^{2}-4x$)\\
\hline
-$v_6$ & $\boldsymbol{x^{6}-7x^{5}+19x^{4}-25x^{3}+16x^{2}-4x}$\\
\hline
\end{longtable}

\begin{center}
Table 6: Computation of chromatic polynomial for the graph in Figure 6
\end{center}

\subsection{Domination polynomial}

For a given graph $G = (V, E)$, below are the required notations for the domination polynomial \cite{dom}. \

\textbf{Open neighborhood} ($N_{G}(v)$) : $\forall \hspace{1mm} v \in V$, $N(v)$ = \{$u \in V$ $|$ \{$u,v$\} $\in E$ \} \ 

\textbf{Closed neighborhood} ($N_{G}[v]$) : $N_{G}(v) \cup \{v\}$ \\

A vertex subset $V' \subseteq V$ is termed a dominating set if $N_{G}[V'] = V$. The dominating set with the minimum cardinality is defined as the minimum dominating set, and the cardinality of such a minimum dominating set is described as the domination number, which is denoted by $\gamma(G)$.  

\begin{definition}
For a given graph $G = (V, E)$, the domination polynomial in terms of dominating sets and domination number is defined as \[ D(G,x) = \sum_{j=\gamma(G)}^{|V|} d_j(G)x^{j} \] \ 
where the coefficients $d_j(G)$ gives the total number of dominating sets with cardinality $j$,  $\forall \hspace{1mm} j \in \{\gamma(G),\cdots,|V|\}$.\
\end{definition}

For the graph mentioned in Figure 6, below are the values for the maximum dominating set(s), domination number, and domination polynomial. \\

\textbf{Minimum dominating sets}: $\{v_1,v_4\},\{v_1,v_5\},\{v_1,v_6\},\{v_3,v_4\},\{v_3,v_5\},\{v_3,v_6\},\{v_2,v_4\},\{v_2,v_5\},\{v_2,v_6\}$ 

\textbf{Domination number}: 2 \ 

\textbf{Domination polynomial} : $9x^{2}+18x^{3}+15x^{4}+6x^{5}+x^{6}$ \

\begin{lemma}
 For a graph G = (V, E), if $Y \subseteq V$ is non-dominating, then $V \setminus Y$ is dominating in $\overline{G}$.
\end{lemma}

\begin{proof}
Since Y is non-dominating in G, it implies that $\exists \hspace{0.6mm} u \in V \setminus Y$ such that $u$ is not connected to any vertex of $Y$ in G. Hence, there is an edge between all vertices of $Y$ and $u$ in $\overline{G}$. So in $\overline{G}$, all vertices of $Y$ are dominated by $u \in V \setminus Y$ and every such vertex of $V$ belonging to $V \setminus Y$.
\end{proof}

For the computation of the domination polynomial based on above mentioned algorithmic approach, below are the required mappings:\\

\begin{tabular}{|p{4.5cm}|p{11.5cm}|}
\hline
 \begin{center}
     Initial state ($Z_0$) 
 \end{center}&  \begin{center}
     $\{([\phi,\phi,\phi],1)\}$
 \end{center}\\
\hline
\begin{center}
    Vertex addition(+$v$) 
\end{center}& \begin{center}
    $\delta_{v}(([D,E,F],f)) = \{([D+v,E,F],f),([D,E,F+v],xf)\} $ 
\end{center}  \\
\hline

\vspace{6mm}

\begin{center}
    Vertex deletion(-$v$) 
\end{center}
&  

\begin{center}
    $\eta_{v}(([D,E,F],f)) = 
    \begin{cases}
      \{([D,E-v,F-v],f)\}, & \text{if}\ v \notin D \\
      \phi, & \text{otherwise}
    \end{cases} $

\end{center}
    \vspace{9mm}
    \\
\hline

\vspace{9mm}
\begin{center}
    Edge insertion($e=\{u,v\}$)
\end{center}
    &
    \begin{center}
$\beta_{e}(([D,E,F],f)) = 
    \begin{cases}
      \{([D-u,E+u,F],f)\}, & \text{if}\ v \in F \hspace{1mm} \text{and} \hspace{1mm} u \in D  \\
      \{([D-v,E+v,F],f)\}, & \text{if}\ u \in F \hspace{1mm} and \hspace{1mm} v \in D  \\
      \{([D,E,F],f)\}, & \text{otherwise}
    \end{cases} $ \end{center} 

 \vspace{9mm}
    \\
\hline
\end{tabular}

\begin{center}
Table 7: Maps for computation of domination polynomial
\end{center}
\vspace{2mm}

\textbf{Explanation}: $y + v$ denotes $y \cup\{v\} $ and $y- v$ represents the deletion of a vertex from the set for y as D or E, or F. Indices are in the form of sets. In the case of the domination polynomial, three scenarios are observed. Firstly, the vertex neither belongs to the dominating set nor the vertex is adjacent to any vertex from the dominating set. Such a vertex is an uncovered vertex. In the second case, the vertex does not belong to the dominating set but is adjacent to a vertex in the dominating set, and such a vertex is a covered vertex. The last case involves a vertex that belongs to the dominating set, and the vertex is a dominating vertex. Indices are considered in the form of a triple consisting of three sets, one for each of the categories. \ 

Since the domination polynomial of a null graph is 1, thereby gives the initialized state as ($[\phi,\phi,\phi],1$). Two states $z_1, z_2$ having identical index $y$ given by $(y,f_1)$ and $(y,f_2)$ respectively is replaced by $(y,f_1+f_2)$.

\begin{longtable}{|p{1.5cm}|p{14cm}|}
\hline
Step & States \\
\hline
init & ([$\phi$,$\phi$,$\phi$],1) \\
\hline
+$v_1$ & ([\{$v_1$\},$\phi$,$\phi$],1),([$\phi$,$\phi$,\{$v_1$\}],$x$)\\
\hline
+$v_3$ & ([\{$v_1,v_3$\},$\phi$,$\phi$],1),([\{$v_1$\},$\phi$,\{$v_3$\}],$x$),([\{$v_3$\},$\phi$,\{$v_1$\}],$x$),([$\phi$,$\phi$,\{$v_1,v_3$\}],$x^{2}$)\\
\hline
$\{v_1,v_3\}$ & ([\{$v_1,v_3$\},$\phi$,$\phi$],1),([$\phi$,\{$v_1$\},\{$v_3$\}],$x$),([$\phi$,\{$v_3$\},\{$v_1$\}],$x$),([$\phi$,$\phi$,\{$v_1,v_3$\}],$x^{2}$)\\

\hline 
+$v_2$ & ([\{$v_1,v_3,v_2$\},$\phi$,$\phi$],1),([\{$v_1,v_3$\},$\phi$,\{$v_2$\}],$x$),([\{$v_2$\},\{$v_1$\},\{$v_3$\}],$x$),([$\phi$,\{$v_1$\},\{$v_3,v_2$\}],$x^{2}$),\

([\{$v_2$\},\{$v_3$\},\{$v_1$\}],$x$),([$\phi$,\{$v_3$\},\{$v_1,v_2$\}],$x^{2}$),([\{$v_2$\},$\phi$,\{$v_1,v_3$\}],$x^{2}$),([$\phi$,$\phi$,\{$v_1,v_3,v_2$\}],$x^{3}$)\\
\hline

\vspace{0.5mm}

$\{v_1,v_2\}$& ([\{$v_1,v_3,v_2$\},$\phi$,$\phi$],1),([\{$v_3$\},\{$v_1$\},\{$v_2$\}],$x$),([\{$v_2$\},\{$v_1$\},\{$v_3$\}],$x$),([$\phi$,\{$v_1$\},\{$v_3,v_2$\}],$x^{2}$),\

([$\phi,$\{$v_2$,$v_3$\},\{$v_1$\}],$x$),([$\phi$,\{$v_3$\},\{$v_1,v_2$\}],$x^{2}$),($\phi$,\{$v_2$\},\{$v_1,v_3$\}],$x^{2}$),([$\phi$,$\phi$,\{$v_1,v_3,v_2$\}],$x^{3}$)\\
\hline

\vspace{0.5mm}

$\{v_2,v_3\}$ & ([\{$v_1,v_3,v_2$\},$\phi$,$\phi$],1),([$\phi$,\{$v_1$,$v_3$\},\{$v_2$\}],$x$),([$\phi$,\{$v_1,v_2$\},\{$v_3$\}],$x$),([$\phi$,\{$v_1$\},\{$v_3,v_2$\}],$x^{2}$),\

([$\phi,$\{$v_2$,$v_3$\},\{$v_1$\}],$x$),([$\phi$,\{$v_3$\},\{$v_1,v_2$\}],$x^{2}$),($\phi$,\{$v_2$\},\{$v_1,v_3$\}],$x^{2}$),([$\phi$,$\phi$,\{$v_1,v_3,v_2$\}],$x^{3}$)\\
\hline
-$v_1$ & ([$\phi$,\{$v_3$\},\{$v_2$\}],$x+x^{2}$),([$\phi$,\{$v_2$\},\{$v_3$\}],$x+x^{2}$),([$\phi$, $\phi$,\{$v_3,v_2$\}],$x^{2}+x^{3}$),([$\phi$,\{$v_3,v_2$\},$\phi$],$x$)\\
\hline
-$v_3$ & ([$\phi$,\{$v_2$\},$\phi$],$2x+x^{2}$),([$\phi$,$\phi$,\{$v_2$\}],$x+2x^{2}+x^{3}$)\\
\hline

\vspace{0.5mm}

+$v_4$ & ([\{$v_4$\},\{$v_2$\},$\phi$],$2x+x^{2}$),([$\phi$,\{$v_2$\},\{$v_4$\}],$2x^{2}+x^{3}$),([\{$v_4$\},$\phi$,\{$v_2$\}],$x+2x^{2}+x^{3}$),\ 

([$\phi$,$\phi$,\{$v_2,v_4$\}],$x^{2}+2x^{3}+x^{4}$)\\
\hline

\vspace{0.5mm}

$\{v_2,v_4\}$ & ([\{$v_4$\},\{$v_2$\},$\phi$],$2x+x^{2}$),([$\phi$,\{$v_2$\},\{$v_4$\}],$2x^{2}+x^{3}$),([$\phi$,\{$v_4$\},\{$v_2$\}],$x+2x^{2}+x^{3}$),\ 

([$\phi$,$\phi$,\{$v_2,v_4$\}],$x^{2}+2x^{3}+x^{4}$)\\
\hline
-$v_2$ & ([\{$v_4$\},$\phi$,$\phi$],$2x+x^{2}$),([$\phi$,$\phi$,\{$v_4$\}],$3x^{2}+3x^{3}+x^{4}$),([$\phi$,\{$v_4$\},$\phi$],$x+2x^{2}+x^{3}$),\\
\hline

\vspace{0.5mm}

+$v_5$ & ([\{$v_4$,$v_5$\},$\phi$,$\phi$],$2x+x^{2}$),([\{$v_4$\},$\phi$,\{$v_5$\}],$2x^{2}+x^{3}$),([\{$v_5$\},$\phi$,\{$v_4$\}],$3x^{2}+3x^{3}+x^{4}$),\ 

([$\phi$,$\phi$,\{$v_4,v_5$\}],$3x^{3}+3x^{4}+x^{5}$),([\{$v_5$\},\{$v_4$\},$\phi$],$x+2x^{2}+x^{3}$),([$\phi$,\{$v_4$\},\{$v_5$\}],$x^{2}+2x^{3}+x^{4}$)\\
\hline

\vspace{0.5mm}

$\{v_4,v_5\}$ & ([\{$v_4$,$v_5$\},$\phi$,$\phi$],$2x+x^{2}$),([$\phi$,\{$v_4$\},\{$v_5$\}],$3x^{2}+3x^{3}+x^{4}$), ([$\phi$,$\phi$,\{$v_4,v_5$\}],$3x^{3}+3x^{4}+x^{5}$),\ 

([\{$v_5$\},\{$v_4$\},$\phi$],$x+2x^{2}+x^{3}$),([$\phi$,\{$v_5$\},\{$v_4$\}],$3x^{2}+3x^{3}+x^{4}$)\\
\hline

\vspace{3.9mm}

+$v_6$ & ([\{$v_4$,$v_5$,$v_6$\},$\phi$,$\phi$],$2x+x^{2}$),([\{$v_4$,$v_5$\},$\phi$,
\{$v_6$\}],$2x^{2}+x^{3}$),([\{$v_6$\},\{$v_4$\},\{$v_5$\}],$3x^{2}+3x^{3}+x^{4}$),\ 

([$\phi$,\{$v_4$\},\{$v_5$,$v_6$\}],$3x^{3}+3x^{4}+x^{5}$), ([\{$v_6$\},$\phi$,\{$v_4,v_5$\}],$3x^{3}+3x^{4}+x^{5}$),\ 

([$\phi$,$\phi$,\{$v_4,v_5,v_6$\}],$3x^{4}+3x^{5}+x^{6}$),([\{$v_5,v_6$\},\{$v_4$\},$\phi$],$x+2x^{2}+x^{3}$),\ 

([\{$v_5$\},\{$v_4$\},\{$v_6$\}],$x^{2}+2x^{3}+x^{4}$),([\{$v_6$\},\{$v_5$\},\{$v_4$\}],$3x^{2}+3x^{3}+x^{4}$),\ 

([$\phi$,\{$v_5$\},\{$v_4,v_6$\}],$3x^{3}+3x^{4}+x^{5}$)\\
\hline

\vspace{3.2mm}
$\{v_4,v_6\}$ & ([\{$v_4$,$v_5$,$v_6$\},$\phi$,$\phi$],$2x+x^{2}$),([\{$v_5$\},\{$v_4$\},\{$v_6$\}],$3x^{2}+3x^{3}+x^{4}$),([\{$v_6$\},\{$v_4$\},\{$v_5$\}],$3x^{2}+3x^{3}+x^{4}$),\ 

([$\phi$,\{$v_4$\},\{$v_5$,$v_6$\}],$3x^{3}+3x^{4}+x^{5}$), ([$\phi$,\{$v_6$\},\{$v_4,v_5$\}],$3x^{3}+3x^{4}+x^{5}$),\ 

([$\phi$,$\phi$,\{$v_4,v_5,v_6$\}],$3x^{4}+3x^{5}+x^{6}$),([\{$v_5,v_6$\},\{$v_4$\},$\phi$],$x+2x^{2}+x^{3}$),\ 

([$\phi$,\{$v_5,v_6$\},\{$v_4$\}],$3x^{2}+3x^{3}+x^{4}$), ([$\phi$,\{$v_5$\},\{$v_4,v_6$\}],$3x^{3}+3x^{4}+x^{5}$)\\
\hline

\vspace{3.3mm}

-$v_4$ & ([\{$v_5$\},$\phi$,\{$v_6$\}],$3x^{2}+3x^{3}+x^{4}$),([\{$v_6$\},$\phi$,\{$v_5$\}],$3x^{2}+3x^{3}+x^{4}$),\ 

([$\phi$,$\phi$,\{$v_5$,$v_6$\}],$3x^{3}+6x^{4}+4x^{5}+x^{6}$), ([$\phi$,\{$v_6$\},\{$v_5$\}],$3x^{3}+3x^{4}+x^{5}$),\ 

([\{$v_5,v_6$\},$\phi$,$\phi$],$x+2x^{2}+x^{3}$),([$\phi$,\{$v_5,v_6$\},$\phi$],$3x^{2}+3x^{3}+x^{4}$), \ 

([$\phi$,\{$v_5$\},\{$v_6$\}],$3x^{3}+3x^{4}+x^{5}$)\\
\hline

\vspace{2.5mm}
$\{v_5,v_6\}$ & ([$\phi$,\{$v_5$\},\{$v_6$\}],$3x^{2}+6x^{3}+4x^{4}+x^{5}$),([$\phi$,\{$v_6$\},\{$v_5$\}],$3x^{2}+6x^{3}+4x^{4}+x^{5}$),\ 

([$\phi$,$\phi$,\{$v_5$,$v_6$\}],$3x^{3}+6x^{4}+4x^{5}+x^{6}$), ([\{$v_5,v_6$\},$\phi$,$\phi$],$x+2x^{2}+x^{3}$),\ 

([$\phi$,\{$v_5,v_6$\},$\phi$],$3x^{2}+3x^{3}+x^{4}$)\\
\hline
-$v_5$ & ([$\phi$,$\phi$,\{$v_6$\}],$3x^{2}+9x^{3}+10x^{4}+5x^{5}+x^{6}$),([$\phi$,\{$v_6$\},$\phi$],$6x^{2}+9x^{3}+5x^{4}+x^{5}$)\\
\hline
-$v_6$ & ([$\phi$,$\phi$,$\phi$],$\boldsymbol{9x^{2}+18x^{3}+15x^{4}+6x^{5}+x^{6}}$)\\
\hline
\end{longtable}

\begin{center}
Table 8: Computation of domination polynomial for the graph in Figure 6
\end{center}

\subsection{Bipartition polynomial}
For a given graph $G = (V,E)$, below are the required notations for the bipartition polynomial \cite{dod2015bipartition}.\

\vspace{1mm}

\textbf{Edge boundary} ($\partial V'$): The edge boundary denoted as $\partial V'$ for a vertex subset $V' \subseteq V$ is \[ \partial V' = \{\{a,b\}\hspace{1mm} | \hspace{1mm} a \in V' and \hspace{1mm} b \in V \setminus V'\} \]

\begin{definition}
For a given graph $G = (V,E)$, bipartition polynomial in terms of edge boundary and open neighbourhood is defined as \[ B(G;x,y,z) = \sum_{V'\subseteq V} x^{|V'|} \sum_{E'\subseteq \partial V'} y^{|N_{(V,E')}(V')|} z^{|E'|} \] 

\end{definition}

Bipartition polynomial provides a generalization for various types of polynomials, including independence and domination polynomial \cite{ok2017bipartition}. Moreover, from the definition, it can be concluded that the sum of all coefficients of $x^{i}y^{j}z^{k}$ for $ \hspace{1mm} 0 \leq i,j \leq |V|$ and $ 0 \leq k \leq |E|$ gives the total number of bipartite subgraphs for $G$.

\begin{center}

    \begin{tikzpicture}  
  [scale=1.8,auto=center,every node/.style={circle,fill=transparent!20}] 
    
  \node (a1) at (-1.732/2,1/2) {$v_1$};  
  \node (a2) at (0,0)  {$v_2$}; 
  \node (a3) at (-1.732/2,-1/2)  {$v_3$};

  \draw (a1) -- (a2); 
  \draw (a2) -- (a3);  
  \draw (a1) -- (a3);

\end{tikzpicture}  
\end{center}

\begin{center}
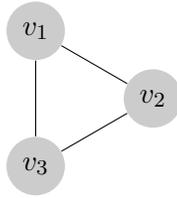
Figure 7: Graph representation with 3 vertices and 3 edges
\end{center}

For the graph mentioned in Figure 7, below are the values for composition order as well as the bipartition polynomial.\

\textbf{Composition order:} ($+v_1,+v_3,\{v_1,v_3\},+v_2,\{v_1,v_2\},\{v_2,v_3\},-v_1,-v_3,-v_2$)

\textbf{Bipartition polynomial:} $1+3x+3x^{2}+x^{3}+6xyz+3xy^{2}z^{2}+6x^{2}yz+3x^{2}yz^{2}$\\ 
\begin{center}
\begin{tikzpicture}  
  [scale=1.5,auto=center,every node/.style={circle,fill=transparent!20}] 
    
  \node (a1) at (0,1) {$v_1$};  
  \node (a2) at (0,0)  {$v_2$}; 
  \node (a3) at (0.865,0.5)  {$v_3$};  
  
  \node (a4) at (2,1) {$v_1$};  
  \node (a5) at (2,0)  {$v_2$}; 
  \node (a6) at (2.865,0.5)  {$v_3$};  
  \node (a7) at (4,1) {$v_1$};  
  \node (a8) at (4,0)  {$v_3$}; 
  \node (a9) at (4.865,0.5)  {$v_2$}; 
  \node (a10) at (6,1) {$v_1$};  
  \node (a11) at (6,0)  {$v_3$}; 
  \node (a12) at (6.865,0.5)  {$v_2$}; 
  \node (a13) at (8,1) {$v_2$};  
  \node (a14) at (8,0)  {$v_3$}; 
  \node (a15) at (8.865,0.5)  {$v_1$};
  \node (a16) at (10,1) {$v_2$};  
  \node (a17) at (10,0)  {$v_3$}; 
  \node (a18) at (10.865,0.5)  {$v_1$};

  \draw (a2) -- (a3);  
  \draw (a4) -- (a6);
   \draw (a8) -- (a9);
   \draw (a10) -- (a12);
   \draw (a14) -- (a15);
   \draw (a16) -- (a18);
\end{tikzpicture}  
\end{center}
\begin{center}
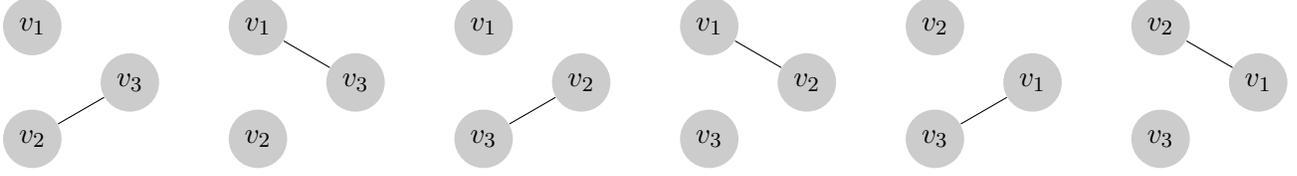
Figure 8: Six instances of bipartite subgraphs representing the term $x^{2}yz$ \

\end{center}
\
\begin{center}
\begin{tikzpicture}  
  [scale=1.8,auto=center,every node/.style={circle,fill=transparent!20}] 
    
  \node (a1) at (0,1) {$v_1$};  
  \node (a2) at (0,0)  {$v_2$}; 
  \node (a3) at (-0.865,0.5)  {$v_3$};  
  
  \node (a4) at (-2,1) {$v_3$};  
  \node (a5) at (-2,0)  {$v_2$}; 
  \node (a6) at (-2.865,0.5)  {$v_1$};  
  \node (a7) at (-4,1) {$v_1$};  
  \node (a8) at (-4,0)  {$v_3$}; 
  \node (a9) at (-4.865,0.5)  {$v_2$};

  \draw (a2) -- (a3);  
  \draw (a1) -- (a3);
  \draw (a4) -- (a6);
   \draw (a5) -- (a6);
  \draw (a7) -- (a9);
  \draw (a8) -- (a9);
\end{tikzpicture}  
\end{center}
\begin{center}
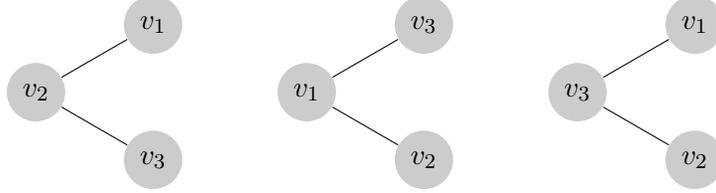
Figure 9: Three instances of bipartite subgraphs representing the term $xy^{2}z^{2}$
\end{center}

\begin{lemma}
For a given graph G = (V,E), $|E| = \frac{1}{2}[xyz]B(G;x,y,z)$, where $[xyz]$ represents the coefficient of $xyz$ in the bipartition polynomial of G.
\end{lemma}

\begin{proof}
The coefficient of $xyz$ in $B(G;x,y,z)$ gives the number of bipartite subgraphs $G'= (Y\cup Z,E')$  with $|Y|=|Z|=|E'|=1$. It implies that for each edge, one endpoint is in $Y$ and the other in $Z$. But each edge is considered exactly twice for the same edge, as the position of $Y$ and $Z$ can be swapped in the polynomial. It gives the number of edges for the original graph $G$ as $\frac{1}{2}[xyz]B(G;x,y,z)$.  
\end{proof}

For the computation of the bipartition polynomial based on above mentioned algorithmic approach, below are the required mappings:\\

\begin{tabular}{|p{4.5cm}|p{11cm}|}
\hline
\begin{center}
    Initial state ($Z_0$)
\end{center}
  &  \begin{center}
      $\{([\phi,\phi,\phi],1)\}$
  \end{center}\\
\hline
\begin{center}
    Vertex addition(+$v$)
\end{center}
 & \begin{center}
     $\delta_{v}(([D,E,F],f)) = \{([D,E+v,F],f),([D,E,F+v],xf)\} $   
 \end{center}\\
\hline
\begin{center}
    Vertex deletion(-$v$) 
\end{center}
& \begin{center}
    $\eta_{v}((y,f)) = \{(y-v,f)\} $
\end{center}\\
\hline
\vspace{17mm}
\begin{center}
    Edge insertion($e=\{u,v\}$)
\end{center}
    &
    \footnotesize
\begin{center}
    $\beta_{e}(([D,E,F],f)) = 
    \begin{cases}
      \{([D,E,F],f),([D+v,E-v,F],yzf)\}, & \\
      \hspace{10mm} \text{if}\ u \in F \hspace{1mm} \text{and} \hspace{1mm} v \in E  \\
      \{([D,E,F],f),([D+u,E-u,F],yzf)\}, & \\
      \hspace{10mm} \text{if}\ v \in F \hspace{1mm} \text{and} \hspace{1mm} u \in E  \\
      \{([D,E,F],f),([D,E,F],zf)\}, & \\
      \hspace{1mm} \text{if either }\ (u \in F \hspace{1mm} \text{and} \hspace{1mm} v \in D) \hspace{1mm} \text{or} \hspace{1mm} (u \in D \hspace{1mm} \text{and} \hspace{1mm}v \in F)  \\
      \{([D,E,F],f)\}, \hspace{1mm} otherwise & 
    \end{cases}$  

\end{center}

 \vspace{12mm}
    \\
    
\hline

\end{tabular}

\begin{center}
\scriptsize \textbf{Table 9:} Maps for computation of bipartition polynomial
\end{center}

\

\textbf{Explanation}: $y + v$ denotes $y \cup\{v\} $ and $y- v$ represents the deletion of a vertex from the set for y as D or E, or F. Since the bipartition polynomial is a polynomial in three variables, we consider as many sets for computation. The first set is used for representing the outer sum, i.e., a subset of the original vertex set. The second set illustrates the vertices that are connected to the vertices of the vertex subset. The final set is required to store those specific vertices sharing an edge between the set from the edge boundary to the original vertex subset. Indices are considered in the form of a triple consisting of three sets, one for each of the categories. Since the bipartition polynomial of a null graph is 1, thereby gives the initialized state as ($[\phi,\phi,\phi],1$). Two states $z_1,z_2$ having identical index $y$ given by $(y,f_1)$ and $(y,f_2)$ respectively is replaced by $(y,f_1+f_2)$.\\ 

\begin{longtable}{|p{1.5cm}|p{15cm}|}
\hline
Step & States \\
\hline
\endfirsthead

\hline
Step & States \\
\hline
\endhead

\hline
\multicolumn{2}{r}{\textit{}} \\
\endfoot

\hline
\endlastfoot

init & ([$\phi$,$\phi$,$\phi$],1) \\
\hline
+$v_1$ & ([$\phi$,\{$v_1$\},$\phi$],1),([$\phi$,$\phi$,\{$v_1$\}],$x$)\\
\hline
+$v_3$ & ([$\phi$,\{$v_1,v_3$\},$\phi$],1),([$\phi$,\{$v_1$\},\{$v_3$\}],$x$),([$\phi$,\{$v_3$\},\{$v_1$\}],$x$),([$\phi$,$\phi$,\{$v_1,v_3$\}],$x^{2}$)\\
\hline
$\{v_1,v_3\}$ & ([$\phi$,\{$v_1,v_3$\},$\phi$],1),([$\phi$,\{$v_1$\},\{$v_3$\}],$x$),([\{$v_1$\},$\phi$,\{$v_3$\}],$xyz$),([$\phi$,\{$v_3$\},\{$v_1$\}],$x$),([\{$v_3$\},$\phi$,\{$v_1$\}],$xyz$),\ 
([$\phi$,$\phi$,\{$v_1,v_3$\}],$x^{2}$)\\
\hline
+$v_2$ & ([$\phi$,\{$v_1,v_3,v_2$\},$\phi$],1),([$\phi$,\{$v_1,v_3$\},\{$v_2$\}],$x$),([$\phi$,\{$v_1,v_2$\},\{$v_3$\}],$x$),([$\phi$,\{$v_1$\},\{$v_3,v_2$\}],$x^{2}$),\
([\{$v_1$\},\{$v_2$\},\{$v_3$\}],$xyz$),([\{$v_1$\},$\phi$,\{$v_3,v_2$\}],$x^{2}yz$),([$\phi$,\{$v_3,v_2$\},\{$v_1$\}],$x$),([$\phi$,\{$v_3$\},\{$v_1,v_2$\}],$x^{2}$),\ 
([\{$v_3$\},\{$v_2$\},\{$v_1$\}],$xyz$),([\{$v_3$\},$\phi$,\{$v_1,v_2$\}],$x^{2}yz$),([$\phi$,\{$v_2$\},\{$v_1,v_3$\}],$x^{2}$),([$\phi$,$\phi$,\{$v_1,v_3,v_2$\}],$x^{3}$)\\
\hline
$\{v_1,v_2\}$& ([$\phi$,\{$v_1,v_3,v_2$\},$\phi$],1),([$\phi$,\{$v_1,v_3$\},\{$v_2$\}],$x$),([\{$v_1$\},\{$v_3$\},\{$v_2$\}],$xyz$),([$\phi$,\{$v_1,v_2$\},\{$v_3$\}],$x$),\ 
([$\phi$,\{$v_1$\},\{$v_3,v_2$\}],$x^{2}$),([\{$v_1$\},$\phi$,\{$v_3,v_2$\}],$2x^{2}yz+x^{2}yz^{2}$),([\{$v_1$\},\{$v_2$\},\{$v_3$\}],$xyz$),\ 
([$\phi$,\{$v_3,v_2$\},\{$v_1$\}],$x$),([\{$v_2$\},\{$v_3$\},\{$v_1$\}],$xyz$),([$\phi$,\{$v_3$\},\{$v_1,v_2$\}],$x^{2}$),([\{$v_3$\},\{$v_2$\},\{$v_1$\}],$xyz$),\ 
([\{$v_3,v_2$\},$\phi$,\{$v_1$\}],$xy^{2}z^{2}$),([\{$v_3$\},$\phi$,\{$v_1,v_2$\}],$x^{2}yz$),([$\phi$,\{$v_2$\},\{$v_1,v_3$\}],$x^{2}$),\  
([\{$v_2$\},$\phi$,\{$v_1,v_3$\}],$x^{2}yz$),([$\phi$,$\phi$,\{$v_1,v_3,v_2$\}],$x^{3}$)\\
\hline
$\{v_2,v_3\}$ & ([$\phi$,\{$v_1,v_3,v_2$\},$\phi$],1),([$\phi$,\{$v_1,v_3$\},\{$v_2$\}],$x$),([\{$v_3$\},\{$v_1$\},\{$v_2$\}],$xyz$),([\{$v_1$\},\{$v_3$\},\{$v_2$\}],$xyz$),\ 
([\{$v_1,v_3$\},$\phi$,\{$v_2$\}],$xy^{2}z^{2}$),([$\phi$,\{$v_1,v_2$\},\{$v_3$\}],$x$),([\{$v_2$\},\{$v_1$\},\{$v_3$\}],$xyz$),([$\phi$,\{$v_1$\},\{$v_3,v_2$\}],$x^{2}$)
([\{$v_1$\},$\phi$,\{$v_3,v_2$\}],$2x^{2}yz+x^{2}yz^{2}$),([\{$v_1$\},\{$v_2$\},\{$v_3$\}],$xyz$),([\{$v_1,v_2$\},$\phi$,\{$v_3$\}],$xy^{2}z^{2}$)
([$\phi$,\{$v_3,v_2$\},\{$v_1$\}],$x$),([\{$v_2$\},\{$v_3$\},\{$v_1$\}],$xyz$),([$\phi$,\{$v_3$\},\{$v_1,v_2$\}],$x^{2}$),\ 
([\{$v_3$\},$\phi$,\{$v_1,v_2$\}],$2x^{2}yz+x^{2}yz^{2}$),([\{$v_3$\},\{$v_2$\},\{$v_1$\}],$xyz$), ([\{$v_3,v_2$\},$\phi$,\{$v_1$\}],$xy^{2}z^{2}$),\ 
([$\phi$,\{$v_2$\},\{$v_1,v_3$\}],$x^{2}$),([\{$v_2$\},$\phi$,\{$v_1,v_3$\}],$x^{2}yz$),([\{$v_2$\},$\phi$,\{$v_1,v_3$\}],$x^{2}yz+x^{2}yz^{2}$),
([$\phi$,$\phi$,\{$v_1,v_3,v_2$\}],$x^{3}$)\\
\hline 
$-v_1$ & ([$\phi$,\{$v_3,v_2$\},$\phi$],$1+x$),([$\phi$,\{$v_3$\},\{$v_2$\}],$x+xyz+x^{2}$),([\{$v_3$\},$\phi$,\{$v_2$\}],$xyz+2x^{2}yz+x^{2}yz^{2}+xy^{2}z^{2}$),\ 
([$\phi$,\{$v_2$\},\{$v_3$\}],$x+xyz+x^{2}$),([\{$v_2$\},$\phi$,\{$v_3$\}],$xyz+xy^{2}z^{2}+2x^{2}yz+x^{2}yz^{2}$),\ 
([$\phi$,$\phi$,\{$v_3,v_2$\}],$x^{2}+2x^{2}yz+x^{2}yz^{2}+x^{3}$),([\{$v_2$\},\{$v_3$\},$\phi$],$xyz$),  ([\{$v_3$\},\{$v_2$\},$\phi$],$xyz+xy^{2}z^{2}$)\\ 
\hline
$-v_3$ & ([$\phi$,\{$v_2$\},$\phi$],$1+2x+2xyz+x^{2}+xy^{2}z^{2}$),([$\phi$,$\phi$,\{$v_2$\}],$x+2xyz+2x^{2}+4x^{2}yz+2x^{2}yz^{2}+xy^{2}z^{2}+x^{3}$),\ 
([\{$v_2$\},$\phi$,$\phi$],$2xyz+xy^{2}z^{2}+2x^{2}yz+x^{2}yz^{2}$),\\ 
\hline
$-v_2$ & ([$\phi$,$\phi$,$\phi$],$\boldsymbol{ 1+3x+3x^{2}+x^{3}+6xyz+3xy^{2}z^{2}+6x^{2}yz+3x^{2}yz^{2}}$)\\
\hline 
\end{longtable}

\begin{center}
Table 10: Computation of bipartition polynomial for the graph in Figure 7
\end{center}

\section{Python Implementation}\label{section-Python}

The code for this project is available on GitHub: 
\href{https://github.com/MehulBafna/Graph-Polynomials}{https://github.com/MehulBafna/Graph-Polynomials}

\normalsize{
\bibliographystyle{unsrt}
\bibliography{references}} 

\begin{thebibliography}{10}

\bibitem{ROBERTSON198339}
N.~Robertson and P.D. Seymour.
\newblock Graph minors. {I}. {Excluding} a forest.
\newblock {\em Journal of Combinatorial Theory, Series B}, 35(1):39--61, 1983.

\bibitem{bodlaender2008combinatorial}
H.~L. Bodlaender and A.~M. Koster.
\newblock Combinatorial {Optimization} on {Graphs} of {Bounded} {Treewidth}.
\newblock {\em The Computer Journal}, 51(3):255--269, 2008.

\bibitem{Bodlaender_1993}
H.~L. Bodlaender.
\newblock A {Tourist} {Guide} through {Treewidth}.
\newblock {\em Acta Cybernetica}, 11(1-2):1--21, 1993.

\bibitem{BODLAENDER19981}
H.~L. Bodlaender.
\newblock A partial k-arboretum of graphs with bounded treewidth.
\newblock {\em Theoretical Computer Science}, 209(1):1--45, 1998.

\bibitem{ponitz2003methode}
A.~P\"onitz.
\newblock {\"Uber} eine {Methode} zur {Konstruktion} von {Algorithmen} f{\"u}r
  die {Berechnung} von {Invarianten} in endlichen ungerichteten {Hypergraphen}.
\newblock {\em PhD thesis. Technischen Universit{\"a}t Bergakademie Freiberg},
  2003.

\bibitem{ind}
I.~Gutman and F.~Harary.
\newblock Generalizations of the matching polynomial.
\newblock {\em Utilitas Mathematica}, 24(1):97--106, 1983.

\bibitem{birkhoff1912determinant}
G.D. Birkhoff.
\newblock A {Determinant} {Formula} for the {Number} of {Ways} of {Coloring} a
  {Map}.
\newblock {\em The Annals of Mathematics}, 14(1/4):42--46, 1912.

\bibitem{dom}
S.~Alikhani.
\newblock Dominating sets and domination polynomials of graphs: {Domination}
  polynomial: A new graph polynomial.
\newblock {\em LAMBERT Academic Publishing}, 2012.

\bibitem{dod2015bipartition}
M.~Dod, T.~Kotek, J.~Preen, and P.~Tittmann.
\newblock Bipartition {Polynomials}, the {Ising} {Model}, and {Domination} in
  {Graphs}.
\newblock {\em Discussiones Mathematicae: Graph Theory}, 35(2), 2015.

\bibitem{poly}
P.~Tittmann.
\newblock {\em Graph {Polynomials}: {The} {Eternal} {Book}}.
\newblock Unpublished manuscript, Hochschule Mittweida, 2021.

\bibitem{ok2017bipartition}
S.~Ok and P.~Tittmann.
\newblock The {Bipartition} {Polynomial} of a {Graph}: {Reconstruction},
  {Decomposition}, and {Applications}.
\newblock {\em arXiv preprint arXiv:1702.03546}, 2017.

\end{thebibliography}

\end{document}